\documentclass[
preprint,
showpacs,
nofootinbib
]{revtex4-1}
\usepackage{amsthm}
\theoremstyle{plain}
\newtheorem{theorem}{Theorem}
\newtheorem*{theorem*}{Theorem}

\newtheorem*{definition*}{Definition}
\newtheorem{lemma}[theorem]{Lemma}
\newtheorem*{lemma*}{Lemma}
\newtheorem{corollary}[theorem]{Corollary}

\usepackage{braket}
\usepackage{delarray}
\usepackage{here}

\usepackage{amsmath}
\usepackage{txfonts}

\usepackage{bm}
\usepackage{color}
\usepackage{ulem}

\newcommand{\dv}[2]{\frac{\mathrm d #1}{\mathrm d #2}}
\newcommand{\Tr}{\operatorname{Tr}}

\newcommand{\ba}{\begin{array}}
\newcommand{\ea}{\end{array}}
\newcommand{\bmat}{\left(\begin{array}}
\newcommand{\emat}{\end{array}\right)}
\newcommand{\no}{\nonumber}

\newcommand{\be}{\begin{eqnarray}}
\newcommand{\ee}{\end{eqnarray}}

\begin{document}
\title{Existence of a Phase Transition in the One-Dimensional Ising Spin Glass Model with Long-Range Interactions on the Nishimori Line}
\author{Manaka Okuyama$^{1,2}$\footnote{corresponding author: manaka.okuyama.d2@tohoku.ac.jp}}
\author{Masayuki Ohzeki$^{1,3,4,5}$}
\affiliation{$^1$Graduate School of Information Sciences, Tohoku University, Sendai 980-8579, Japan}
\affiliation{$^2$School of Computing, Institute of Science Tokyo, Tokyo 152-8551, Japan}
\affiliation{$^3$Department of Physics, Institute of Science Tokyo, Tokyo 152-8551, Japan}
\affiliation{$^4$Research and Education Institute for Semiconductors and Informatics, Kumamoto University, Kumamoto 860-0862, Japan}
\affiliation{$^5$Sigma-i Co., Ltd., Tokyo 108-0075, Japan} 

\begin{abstract}  
Dyson [Commun. Math. Phys. \textbf{12}, 91 (1969)] rigorously proved the existence of a phase transition in the one-dimensional Ising model with long-range interactions of the form $r^{-\alpha}$ for $1 < \alpha < 2$.
In the present study, we extend this result to the Ising spin glass model with Gaussian disorder on the Nishimori line.
Following Dyson's method, we first prove the existence of long-range order at finite low temperatures in the Dyson hierarchical Ising spin glass model on the Nishimori line, with power-law-like interactions $J(r) \sim r^{-\alpha}$ for $1 < \alpha < 3/2$.
The key ingredients of the proof are the interpolation method developed in the rigorous analysis of mean-field spin glass models, the Gibbs--Bogoliubov inequality on the Nishimori line, and the Tsirelson--Ibragimov--Sudakov inequality (Gaussian concentration inequality).
We then use the Griffiths inequality on the Nishimori line to rigorously establish the existence of a phase transition in the one-dimensional Ising spin glass model with long-range interactions on the Nishimori line for $1 < \alpha < 3/2$. 
For $3/2 \le \alpha \le 2 $, the existence of a phase transition remains an open problem.

\end{abstract}
\date{\today}
\maketitle

\section{Introduction}

One of the most challenging problems in statistical physics is the rigorous analysis of phase transitions in spin glass models.
The rigorous analysis of mean-field spin glass models has made remarkable progress following the development of the interpolation method~\cite{GT}, and the replica symmetry breaking solution constructed by Parisi~\cite{Parisi} has been mathematically established~\cite{Guerra,Talagrand}.
In contrast, the rigorous analysis of finite-dimensional systems is generally extremely difficult.
One of the few notable exceptions is the so-called Nishimori line~\cite{Nishimori,Nishimori2,Nishimori3}, a special line in the phase diagram of the Ising spin glass model that connects the high-temperature regime with zero-mean random interactions to the low-temperature regime with a strong ferromagnetic bias.
On this line, gauge transformations yield various exact and rigorous results, including exact expressions for the internal energy, upper bounds on the specific heat, correlation identities, and analogues of the Griffiths inequalities~\cite{MNC,Kitatani}.
Nevertheless, even on the Nishimori line, rigorous results on phase transitions in finite-dimensional systems remain scarce~\cite{HM,GS}.
In the present study, we focus on the one-dimensional Ising spin glass model with long-range interactions on the Nishimori line and present a new rigorous result on the existence of a phase transition.

In contrast to one-dimensional spin systems with short-range interactions, one-dimensional spin systems with long-range interactions, in which the interaction strength decays as $J(r)\propto r^{-\alpha}$ with distance $r$, can exhibit a phase transition at finite temperature, depending on the value of $\alpha$.
This was first rigorously proved by Dyson for the one-dimensional Ising model with long-range interactions for $1<\alpha<2$~\cite{Dyson}.
Dyson first established the existence of a phase transition for the Ising model on the Dyson hierarchical lattice, where the interactions decay in a power-law-like form $J(r)\sim r^{-\alpha}$, for $1<\alpha<2$.
He then rigorously proved, using the Griffiths inequality~\cite{Griffiths,KS}, that a phase transition occurs in the one-dimensional Ising model with long-range interactions for the same range of $\alpha$.
For $\alpha=2$, Fr\"{o}hlich and Spencer proved the existence of a phase transition using a contour argument~\cite{FS}.
In addition, it has been shown that for $0\le\alpha<1$, with an appropriate normalization of the interaction, the free energy of the Ising model coincides with that of the corresponding mean-field model~\cite{Mori}.

The rigorous analysis of the one-dimensional Ising model with long-range interactions has also been extended to the case of a random magnetic field.
For $0\le\alpha\le1$, with an appropriate normalization of the interaction, Tsuda and Nishimori proved that the free energy coincides with that of the corresponding mean-field model~\cite{TN}.
Aizenman and Wehr proved that there is no long-range order at any temperature for $\alpha>3/2$~\cite{AW}.
The Imry--Ma argument~\cite{IM} suggests the existence of a phase transition for $1<\alpha<3/2$, and Cassandro, Orlandi, and Picco proved the existence of long-range order for $3-\log 3/\log 2<\alpha<3/2$ using contour methods~\cite{COP}.
The existence of a phase transition in the region $1<\alpha\le3-\log 3/\log 2$ had long remained an open problem, but has recently been resolved affirmatively by Ding, Huang, and Maia through a refinement of contour methods~\cite{DHM}.
In addition, the existence of long-range order in the random-field Ising model on the Dyson hierarchical lattice for $1<\alpha<3/2$ has recently been established by extending Dyson's method using concentration inequalities from probability theory~\cite{OO-Dyson}.

One-dimensional spin systems with long-range interactions have also been extensively studied in the context of the Ising spin glass model~\cite{KS2,EH,KAS,KY,Moore,MG,APT,LPTL}, in which the interactions $J_{ij}$ are Gaussian random variables with zero mean and variance $1/r^{\alpha}$.
One of the main motivations for these studies is that the effective spatial dimension can be tuned by varying $\alpha$, allowing one to investigate the validity of the mean-field picture in finite-dimensional systems.
For $\alpha>2$, it is known that there is no phase transition at any finite temperature~\cite{EH2,Enter}.
For $0\le\alpha<1$, with an appropriate normalization of the interaction, it has been rigorously proved that the free energy coincides with that of the corresponding mean-field model, namely the Sherrington--Kirkpatrick model~\cite{OO-SK}.
For $1<\alpha<2$, it is widely believed that a phase transition occurs, with mean-field-type behavior for $1<\alpha<4/3$ and short-range-type behavior for $4/3<\alpha<2$~\cite{KAS}.
However, to date there are no mathematically rigorous results establishing the existence of a phase transition in the region $1<\alpha<2$.
Compared with the ferromagnetic Ising model and the random-field Ising model, the rigorous analysis of the Ising spin glass model is considerably more difficult.

In the present study, we consider the one-dimensional Ising spin glass model with long-range interactions on the Nishimori line, in which the interactions $J_{ij}$ are Gaussian random variables with mean $\beta/r^{\alpha}$ and variance $1/r^{\alpha}$, where $\beta$ denotes the inverse temperature.
For $0\le\alpha<1$, with an appropriate normalization of the interaction, it has been proved that the free energy coincides with that of the Sherrington--Kirkpatrick model on the Nishimori line~\cite{OO-SK}.
On the other hand, to the best of our knowledge, the thermodynamic properties in the region $\alpha>1$ have not been investigated.
We extend Dyson's method for the random-field Ising model~\cite{OO-Dyson} to the Ising spin glass model by exploiting the special properties of the Nishimori line.
As a result, we rigorously prove the existence of ferromagnetic order in the one-dimensional Ising spin glass model with long-range interactions on the Nishimori line at finite low temperatures for $1<\alpha<3/2$.
We note that the behavior in the region $3/2 \le \alpha \le 2 $ remains an open problem.

The proof consists of the following three steps, following the strategy originally introduced by Dyson for the Ising model~\cite{Dyson}:
\begin{enumerate}
\item First, we prove that the Dyson hierarchical Ising spin glass model on the Nishimori line exhibits long-range order at finite low temperatures for $1<\alpha<3/2$ (Theorem~1).
\item Next, we show that the long-range order in the one-dimensional Ising spin glass model with long-range interactions on the Nishimori line is always greater than or equal to that of the corresponding Dyson hierarchical model (Theorem~2).
This result implies the existence of long-range order in the one-dimensional Ising spin glass model with long-range interactions on the Nishimori line at low temperatures for $1<\alpha<3/2$.
\item Finally, we prove that there is no long-range order at high temperatures in the one-dimensional Ising spin glass model with long-range interactions on the Nishimori line (Theorem~3).
\end{enumerate}

By exploiting the Griffiths inequality on the Nishimori line~\cite{MNC}, Theorems~\ref{theorem2} and~\ref{theorem3} can be proved essentially in the same manner as in the ferromagnetic Ising model~\cite{Dyson}.
In contrast, the proof of Theorem~\ref{theorem1} requires substantial technical modifications compared with the Ising model case, and methods developed in the rigorous analysis of mean-field spin glass models~\cite{AK} play a crucial role.

The remainder of this paper is organized as follows.
In Section~II, we define the one-dimensional Ising spin glass model with long-range interactions and the Dyson hierarchical Ising spin glass model on the Nishimori line, and state the main results.
Section~III is devoted to the proof of Theorem~\ref{theorem1}, which establishes the existence of long-range order in the Dyson hierarchical Ising spin glass model on the Nishimori line.
In Sections~IV and~V, we prove Theorems~\ref{theorem2} and~\ref{theorem3}, respectively, concerning the comparison with the one-dimensional Ising spin glass model with long-range interactions on the Nishimori line and the absence of long-range order at high temperatures.
Finally, in Section~VI, we conclude with a discussion of the results and related open problems.

\section{Model and results}
We consider the one-dimensional Ising spin glass model with long-range interactions decaying as $1/r^\alpha$ on the Nishimori line.
The Hamiltonian is defined by
\be
H_{\text{long}}&=& -\sum_{i<j}  J_{ij}\sigma_i \sigma_j , \label{long-H}
\ee
where $\sigma_i \in \{ -1,1\}$ for all $i\in \mathbb{Z}$, and free boundary conditions are assumed.
The interactions $J_{ij}$ are independent Gaussian random variables given by
\be
J_{ij}&\sim& \mathcal{N}\left(\frac{4\beta} {|i-j|^{\alpha}}, \frac{4}{|i-j|^{\alpha}}\right), \label{Jij-def}
\ee
where $1<\alpha$ and $\beta$ denotes the inverse temperature.
This choice of the Gaussian distribution satisfies the Nishimori line condition~\cite{Nishimori}.
The constant factor 4 in Eq. (\ref{Jij-def}) is not essential and is introduced solely to simplify the proof of Theorem~2.

Here, we briefly summarize some important properties of the Ising spin glass model on the Nishimori line.
On the Nishimori line, correlation functions for any pair of sites $(i,j)$ satisfy the following identity~\cite{Nishimori}:
\be
\mathbb{E}[\langle \sigma_i \sigma_j \rangle] &=& \mathbb{E}[\langle \sigma_i \sigma_j \rangle^2]\ge0 \label{Nishimori-identity},
\ee
where $\mathbb{E}[\cdots]$ denotes the expectation with respect to all quenched random variables, and $\langle\cdots\rangle$ denotes the thermal average for the Ising spin glass model on the Nishimori line.
Next, we consider two systems $C$ and $D$ defined on the same set of spins, with Hamiltonians
\be
H_C &=&- \sum_{\langle i,j\rangle} J_{C,ij}\sigma_i \sigma_j,
\\
H_D &=&- \sum_{\langle i,j\rangle} J_{D,ij}\sigma_i \sigma_j,
\ee
where the interactions in systems $C$ and $D$ are independent Gaussian random variables distributed as
\be
J_{C,ij}\sim \mathcal{N}(\beta c_{ij}, c_{ij}) ,
\\
J_{D,ij}\sim \mathcal{N}(\beta d_{ij}, d_{ij}).
\ee
Assume that for any pair of sites $(i,j)$,
\be
c_{ij} \ge d_{ij} \ge0,
\ee
Then, for any pair of sites $(k,l)$, the corresponding correlation functions satisfy the monotonicity property
\be
\mathbb{E}[\langle \sigma_k \sigma_l \rangle_C] &\ge& \mathbb{E}[\langle \sigma_k \sigma_l \rangle_D], \label{Griffiths-NL}
\ee
which is known as the Griffiths inequality on the Nishimori line~\cite{MNC,Kitatani}.

It is difficult to analyze the one-dimensional Ising spin glass model with long-range interactions on the Nishimori line directly.
Instead, we consider the Dyson hierarchical Ising spin glass model on the Nishimori line.
For each positive integer $N$, the system consists of $2^N$ Ising spins $\sigma_i = \pm1$, labeled by indices $i=1,2,\cdots,2^N$.
The Dyson hierarchical Ising spin glass model on the Nishimori line is defined recursively by the Hamiltonian
\be
H_{N}(\vec{\sigma})&=&H_{N-1}^1(\vec{\sigma}_1) + H_{N-1}^2(\vec{\sigma}_2) -\sum_{i,j=1}^{2^{N}} J_{ij}^{(N)}\sigma_i \sigma_ j  , \label{Dyson-H}
\\
H_0(\vec{\sigma})&=&0 ,
\ee
where $\vec{\sigma}_1\equiv \{ \sigma_i \}_{1\le i \le 2^{N-1}}$, $\vec{\sigma}_2\equiv \{ \sigma_i \}_{2^{N-1}+1\le i \le 2^{N}}$, and the coupling constants $J_{ij}^{(N)}$ are independent Gaussian random variables distributed as
\be
J_{ij}^{(N)} &\stackrel{\text{iid}}{\sim}& \mathcal{N}\left( \frac{\beta b_{N}}{2^{2 N} }, \frac{ b_{N}}{2^{2 N} } \right), \label{Jij-Dyson}
\ee
with 
\be
b_{N}=2^{(2-\alpha)N}.
\ee
This choice of the Gaussian distribution satisfies the Nishimori line condition~\cite{Nishimori}.
The thermal average with respect to the Hamiltonian $H_{N}(\vec{\sigma})$ is defined by
\be
\langle \cdots  \rangle_N &=&\frac{\Tr(\cdots e^{-\beta H_{N}(\vec{\sigma})})}{\Tr(e^{-\beta H_{N}(\vec{\sigma})})},
\ee
where $\Tr$ denotes the summation over all spin configurations.
For each $p=0,1,\cdots,N$ and $r=1,2,\cdots,2^{N-p}$, we define the block spin sum
\be
S_{p,r}&=& \sum_j \sigma_j   , \quad  (r-1)2^p +1 \le j \le r2^p  ,\label{Spr}
\\
S_{p,r}&=& S_{p-1,2r-1}+S_{p-1,2r}.
\ee
We define the long-range order parameter for the Dyson hierarchical Ising spin glass model on the Nishimori line by
\be
0\le f_N(p)&=&\frac{1}{2^{2p}} \mathbb{E}[ \langle S_{p,r}^2  \rangle_N] \le1 . \label{def-fN}
\ee
The Griffiths inequality on the Nishimori line~\eqref{Griffiths-NL} implies that, for fixed \( p \), the sequence \( f_N(p) \) is non-decreasing in \( N \): \
\be
f_N(p) \le f_{N+1}(p) . \label{N-mono}
\ee
On the other hand, the elementary inequality \( x^2 + y^2 \ge 2xy \) implies that, for fixed \( N \), \( f_N(p) \) is non-increasing in \( p \): \
\be
f_{N}(p-1) \ge f_N(p). \label{p-mono}
\ee
Equation~\eqref{N-mono} guarantees the existence of the limit
\be
f(p) \equiv \lim_{N\to\infty} f_N(p). \label{p-N-infy}
\ee
Furthermore, Eq.~\eqref{p-mono} implies
\be
1\ge f(p-1) \ge f(p)  \ge0\label{p-mono-infty}.
\ee
Therefore, the long-range order parameter in the thermodynamic limit is well defined by
\be
m^2 &\equiv& \lim_{p\to\infty}f(p). \label{m^2-def}
\ee
When \( m^2 > 0 \), the system is said to be in the ferromagnetic phase. This order parameter is a natural extension of the one introduced for the ferromagnetic Ising model on the Dyson hierarchical lattice~\cite{Dyson}.

Our first result concerns the existence of long-range order in the Dyson hierarchical Ising spin glass model on the Nishimori line.
\begin{theorem}[]\label{theorem1}
Let $1<\alpha <3/2$ and $\beta \ge 2^{(\alpha-2)/2}$. 
Then, the long-range order parameter of the Dyson hierarchical Ising spin glass model on the Nishimori line satisfies the following lower bound:
\be
m^2&\ge& \frac{1}{2}+\frac{1}{2} \mathbb{E}[  \langle  \sigma_1 \sigma_2  \rangle_1]  -\frac{2^{2\alpha-1}}{\beta^2(4-2^\alpha)} (1+\log \beta) -\frac{4^{-1 +\alpha } (2^\alpha (-4 + \alpha) - 8 (-3 + \alpha)) \log 2 }{\beta^2(4 - 2^\alpha)^2 }
\no\\
&&- \frac{  2^{2 \alpha-3/2 } (2^\alpha (-4 + \alpha) - 2^{5/2} (-3 + \alpha)) R^{1/2}  \log 2 }{\beta (2^{3/2}  - 2^\alpha)^2 }
-\frac{4^\alpha R^{1/2} \log(\beta+\beta\sqrt{2\pi e}  )}{ \beta (4-2^{\alpha+1/2})} ,
\ee
where $R=2/(2^\alpha-2)$.
In particular, the model exhibits long-range order at sufficiently low temperatures.
Moreover, in the zero-temperature limit,
\be
m^2&=&1.
\ee
\end{theorem}

Next, we compare the Dyson hierarchical Ising spin glass model with the one-dimensional Ising spin glass model with long-range interactions.
Owing to the hierarchical structure, one can show that the interactions in the one-dimensional Ising spin glass model with long-range interactions on the Nishimori line are stronger than those in the Dyson hierarchical Ising spin glass model on the Nishimori line.
As a consequence, the Griffiths inequality on the Nishimori line~\eqref{Griffiths-NL} implies the following result.
\begin{theorem}[]\label{theorem2}
For any $\beta$ and $\alpha>1$, the long-range order parameter of the one-dimensional Ising spin glass model with long-range interactions on the Nishimori line is greater than or equal to that of the Dyson hierarchical Ising spin glass model on the Nishimori line.
Consequently, for $1<\alpha <3/2$, the one-dimensional Ising spin glass model with long-range interactions on the Nishimori line exhibits long-range order at finite low temperatures.
\end{theorem}

Finally, we prove the absence of long-range order at high temperatures.
By combining the method developed for the ferromagnetic Ising model~\cite{Dyson} with the approach used in Ref.~\cite{OO-NL-bound}, where an upper bound on the transition temperature on the Nishimori line was derived using the Griffiths inequality on the Nishimori line, we obtain a rigorous proof of the nonexistence of long-range order at high temperatures.
\begin{theorem}[]\label{theorem3}
If
\be
32\beta^2 \sum_{i=1}^{\infty}\frac{1}{ |i|^{\alpha}} <1,
\ee
then the long-range order parameter of the one-dimensional Ising spin glass model with long-range interactions on the Nishimori line vanishes, namely,
\be
m^2=0.
\ee
\end{theorem}

From Theorems~2 and~3, we immediately obtain the following conclusion.
\begin{corollary}[]
For $1<\alpha <3/2$, the one-dimensional Ising spin glass model with long-range interactions on the Nishimori line undergoes a phase transition at a finite temperature.
\end{corollary}

\section{Proof of Theorem \ref{theorem1}}

\subsection{Sketch of proof}
Following Dyson's method~\cite{Dyson}, we rigorously prove that the Dyson hierarchical Ising spin glass model on the Nishimori line exhibits spontaneous magnetization at sufficiently low temperatures.
In particular, our argument closely follows the proof for the random-field Ising model on the Dyson hierarchical lattice developed in Ref.~\cite{OO-Dyson}.

First, Eqs.~\eqref{p-mono}, \eqref{p-N-infy}, and \eqref{m^2-def} imply that
\be
m^2 &\ge& \liminf_{N\to\infty} f_N(N)\label{m^r-f_N} ,
\ee
where $f_N(N)$ is obtained from $f_N(p)$ by setting $p = N$.
Therefore, it suffices to show that $f_N(N)$ admits a strictly positive lower bound.
The essence of Dyson's method is to exploit the hierarchical structure in order to derive a recurrence relation for the long-range order parameter.
More precisely, if one can establish a bound of the form
\be
f_N(N) &\ge& f_{N-1}(N-1) - \frac{1}{\beta} \mathcal{O}\left(2^{-(3/2-\alpha)N}\right),
\ee
then, for $1<\alpha<3/2$, it follows that
\be
f_N(N) &\ge& f_1(1) - \frac{1}{\beta}\mathcal{O}(1),
\ee
which completes the proof of the existence of long-range order at sufficiently low temperatures.

In the case of the ferromagnetic Ising model and the random-field Ising model, the Gibbs--Bogoliubov inequality~\cite{Kuzemsky} allows one to express the difference in the long-range order parameter in terms of a difference of free energies, and Dyson's method applies successfully.
In contrast, for the Ising spin glass model, Dyson's method cannot be applied directly, since the long-range order parameter does not naturally connect to the Gibbs--Bogoliubov inequality due to the presence of quenched randomness.
This difficulty is overcome by employing the Gibbs--Bogoliubov inequality on the Nishimori line~\cite{OO-GB}, which makes it possible to relate the difference in long-range order to a difference of free energies (Lemma~5).

The remaining task is to evaluate the resulting difference of free energies precisely (Lemma~6).
Although the techniques used for the ferromagnetic Ising model~\cite{Dyson} and the random-field Ising model~\cite{OO-Dyson} cannot be applied directly to the Ising spin glass model on the Nishimori line, essentially the same type of inequality can be obtained by means of the interpolation method~\cite{GT}.
In particular, our approach builds upon the rigorous analysis of mean-field spin glass models on the Nishimori line using the Franz--Parisi potential~\cite{AK}.
We note that, in the course of estimating the difference of free energies, it is necessary to invoke Gaussian concentration inequalities in order to control fluctuation terms arising from randomness (Lemma~7).
This requirement restricts our analysis to the range $1<\alpha<3/2$ and prevents us from treating the region $3/2\le\alpha\le2$ (see the discussion in Section~\ref{Discussions}).
With the above procedure, we obtain a recurrence relation for the long-range order parameter (Lemma~8), from which Theorem~1 follows immediately.

\subsection{Proof}
First, we introduce the quenched pressure function of the Dyson hierarchical Ising spin glass model on the Nishimori line, defined by
\be
P_N(x_N,x_{N-1},\cdots,x_1)&=& \mathbb{E}[\log\Tr  e^{-\beta H_{N}(\vec{\sigma})}],
\ee
where $x_N=\beta^2 b_N/2^{2N}$.
The starting point of our analysis is the following recurrence relation.
\begin{lemma}\label{lemma5}
The quantity $f_N(N)$ satisfies the inequality
\be
f_N(N)
&\ge&  f_{N-1}(N-1)+ \frac{2}{\beta^2 b_N} \left(P_N(x_N,x_{N-1},\cdots,x_1)-2P_{N-1}(2x_N+x_{N-1},x_{N-2},\cdots,x_1) \right) .
\ee
\end{lemma}
This recurrence relation follows from the Gibbs--Bogoliubov inequality on the Nishimori line~\cite{OO-GB}; see Sec.~\ref{proof-lemma5} for the proof.

Lemma~\ref{lemma5} requires a precise evaluation of the difference between the two quenched pressure functions $P_N(x_N,x_{N-1},\cdots,x_1)$ and $P_{N-1}(2x_N+x_{N-1},x_{N-2},\cdots,x_1)$.
To this end, we introduce an interpolating Hamiltonian
\be
 H_{N,t}(\vec{\sigma})&=& H_{N-1}^1(\vec{\sigma}_1) + H_{N-1}^2(\vec{\sigma}_2) -\sum_{i,j=1}^{2^{N}} K_{ij}^{(N)}(t)\sigma_i \sigma_ j 
\no\\
&& -\sum_{i,j=1}^{2^{{N-1}}} K_{ij}^{(N-1)}(t)\sigma_i \sigma_ j -\sum_{i,j=2^{N-1}+1}^{2^{{N}}} K_{ij}^{(N-1)}(t)\sigma_i \sigma_ j  ,
\ee
where the coupling constants are independent Gaussian random variables distributed as
\be
K_{ij}^{(N)}(t) &\stackrel{\text{iid}}{\sim} & \mathcal{N}\left(t  \frac{\beta b_{N}}{2^{2 N}} ,t \frac{b_{N}}{2^{2 N}} \right) ,
\\
K_{ij}^{(N-1)}(t) &\stackrel{\text{iid}}{\sim} & \mathcal{N}\left((1-t) \beta \frac{2b_{N}}{2^{2 N}} ,(1-t) \frac{2b_{N}}{2^{2 N}} \right).
\ee
From Eq.~\eqref{Spr}, it follows that
\be
-2^{(N-1)}\le S_{N-1,1} \le 2^{(N-1)}  \label{interval}.
\ee
Let $r_N$ be a positive integer such that
\be
(\beta^2 b_N)^{1/2}< r_N \le 1+(\beta^2 b_N)^{1/2}, \label{lange-r_N}
\ee
and divide the interval~\eqref{interval} into $r_N$ equal subintervals $I_k, k=1,2,\cdots, r_N$~\textsuperscript{\ref{fn:rN-exception}}.
\footnotetext{
Strictly speaking, if $(\beta^2 b_N)^{1/2}>2^{N-1}$, we set $r_N=2^{N-1}+1$, the number of possible values of $S_{N-1,1}$. Then the
intervals may be chosen so that each contains at most one possible value of $S_{N-1,1}$, and hence the estimate below holds trivially. Since $r_N\le 1+(\beta^2 b_N)^{1/2}$ also holds in this case, the subsequent
bounds are unchanged.\label{fn:rN-exception}
}
For $S_{N-1,1}, S_{N-1,2}\in I_k $, we have
\be
(S_{N-1,1} -S_{N-1,2})^2 \le \frac{2^{2N}}{r_N^2}<\frac{2^{2N}}{\beta^2 b_N}. \label{res-S1-S2}
\ee
We define the interpolating restricted pressure function by
\be
Q_N(t)&=& \mathbb{E}[\log( \sum_{k=1}^{r_N} \Tr_{ \{ S_{N-1,1}\in I_k\}} \Tr_{ \{ S_{N-1,2}\in I_k\}}  \exp(-\beta H_{N,t}(\vec{\sigma}))   ) ], \label{def-Q}
\ee
where the sums over spin configurations are restricted to those for which $S_{N-1,1}$ and $S_{N-1,2}$ belong to the same interval $I_k$.
By definition, Eq.~\eqref{def-Q} immediately implies
\be
P_N(x_N,x_{N-1},\cdots,x_1)&\ge&Q_N(1) \label{P->Q}.
\ee
Moreover,
\be
Q_N(0)&=&\mathbb{E}[\log (\sum_{k=1}^{r_N} \Tr_{ \{ S_{N-1,1}\in I_k\}} \Tr_{ \{ S_{N-1,2}\in I_k\}}  \exp(-\beta H_{N,0}(\vec{\sigma}))  ) ]
\no\\
&=&\mathbb{E}[\log (\sum_{k=1}^{r_N} Z_{N,k}(\{J_1,K_1\}) Z_{N,k}(\{J_2,K_2\})) ] ,
\ee
where
\be
Z_{N,k}(\{J_1,K_1\})&\equiv&\Tr_{ \{ S_{N-1,1}\in I_k\}}  \exp(-\beta H_{N-1}^1(\vec{\sigma}_1)   + \beta \sum_{i,j=1}^{2^{{N-1}}} K_{ij}^{(N-1)}(0)\sigma_i \sigma_ j ),
\label{Z_{N,k}}
\\
Z_{N,k}(\{J_2,K_2\})&\equiv&\Tr_{ \{ S_{N-1,2}\in I_k\}}  \exp(-\beta H_{N-1}^2(\vec{\sigma}_2)   + \beta \sum_{i,j=2^{N-1}+1}^{2^{{N}}} K_{ij}^{(N-1)}(0)\sigma_i \sigma_ j  ).
\ee
The interpolating restricted pressure $Q_N(t)$ allows us to derive the following estimate (see Section~\ref{proof-lemma6} for the proof).
\begin{lemma}[] \label{lemma6}
For $\alpha >1$, the following inequality holds:
\be
&&P_N(x_N,x_{N-1},\cdots,x_1)-2P_{N-1}(2x_N+x_{N-1},x_{N-2},\cdots,x_1)
\no\\
&\ge& -1  -\mathbb{E}[\log \sum_{k=1}^{r_N} \frac{Z_{N,k}(\{J_1,K_1\})  } { Z_{N,k}(\{J_2,K_2\})  }]
\no\\
&\ge& -1  -\mathbb{E}[\max _k\log (r_N \frac{Z_{N,k}(\{J_1,K_1\})  } { Z_{N,k}(\{J_2,K_2\})  })] . \label{bottle}
\ee\
\end{lemma}
The last term in Eq.~\eqref{bottle} constitutes the main computational bottleneck.
To control this term, we invoke Gaussian concentration inequalities for Lipschitz functions, specifically the Tsirelson--Ibragimov--Sudakov inequality~\cite{BLM}, which yields the following bound (see Section~\ref{proof-lemma7} for the proof).
\begin{lemma}[]\label{lemma7}
For $\alpha >1$,
\be
\mathbb{E}[ \max_k \log  \frac{Z_{N,k}(\{J_1,K_1\})  } { Z_{N,k}(\{J_2,K_2\})  }]  
&\le&\beta 2^{N/2}R_N^{1/2}\log \left( r_N (1+\sqrt{2\pi e} ) \right),
\ee
where $R_N=(2(1-2^{(1-\alpha)N}))/(2^\alpha-2)$.
\end{lemma}
By combining Lemmas~\ref{lemma5}, \ref{lemma6}, and~\ref{lemma7}, we obtain the desired recurrence relation.
\begin{lemma}\label{lemma8}
Let $N\ge2$. Then the following inequality holds:
\be
f_N(N)  &\ge& f_{N-1}(N-1) - \frac{2}{\beta^2 b_N} \left(1 +(1+\beta 2^{N/2}R_N^{1/2}) \log r_N + \beta 2^{N/2}R_N^{1/2}\log (1+\sqrt{2\pi e}  )  \right).
\ee
\end{lemma}
We are now in a position to prove Theorem \ref{theorem1}.
\begin{proof}[Proof of Theorem 1]
Lemma \ref{lemma8} implies
\be
f_N(N)  &\ge& f_1(1)- \sum_{p=2}^N  \frac{2}{\beta^2 b_p} \left(1 +(1+\beta 2^{p/2}R_p^{1/2}) \log r_p + \beta 2^{p/2}R_p^{1/2}\log (1+\sqrt{2\pi e}  )  \right)
\no\\
&\ge& f_1(1)- \sum_{p=2}^\infty \frac{2}{\beta^2 b_p} \left(1 +(1+\beta 2^{p/2}R^{1/2}) \log r_p + \beta 2^{p/2}R^{1/2}\log (1+\sqrt{2\pi e}  )  \right) ,\label{f_N-inf-sum}
\ee
where 
\be
R = \lim_{p\to\infty}R_p =2/(2^\alpha-2).
\ee
For $\beta \ge 2^{(\alpha-2)/2}$, Eq.~\eqref{lange-r_N} yields 
\be
r_p \le 1+(\beta^2 b_p)^{1/2} \le 2(\beta^2 b_p)^{1/2}.
\ee
Substituting this bound into Eq.~\eqref{f_N-inf-sum}, we obtain the explicit lower bound
\be
f_N(N) &\ge& f_1(1)- \sum_{p=2}^\infty \frac{2}{\beta^2 b_p} \left(1 +(1+\beta 2^{p/2}R^{1/2}) \log (2(\beta^2 b_p)^{1/2}) + \beta 2^{p/2}R^{1/2}\log (1+\sqrt{2\pi e}  )  \right) 
\no\\
&=&f_1(1)- \sum_{p=2}^\infty \frac{2}{\beta^2 b_p} \left(1 +(1+\beta 2^{p/2}R^{1/2}) \left(\log (\beta ) +\log ( 2b_p^{1/2}) \right)+ \beta 2^{p/2}R^{1/2}\log (1+\sqrt{2\pi e}  )  \right) 
\no\\
&=&f_1(1)- \sum_{p=2}^\infty \frac{2}{\beta^2 b_p} \left(1+\log (\beta )  +(1+\beta 2^{p/2}R^{1/2}) \log ( 2b_p^{1/2}) + \beta 2^{p/2}R^{1/2}\log(\beta+\beta\sqrt{2\pi e}  )  \right) 
\no\\
&=&f_1(1)-\frac{2^{2\alpha-1}}{\beta^2(4-2^\alpha)} (1+\log \beta) -\frac{4^{-1 +\alpha } (2^\alpha (-4 + \alpha) - 8 (-3 + \alpha)) \log 2 }{\beta^2(4 - 2^\alpha)^2 }
\no\\
&&-2^{2 \alpha-3/2 }  \frac{(2^\alpha (-4 + \alpha) - 2^{5/2} (-3 + \alpha)) R^{1/2}  \log 2 }{\beta (2^{3/2}  - 2^\alpha)^2 }
-\frac{4^\alpha R^{1/2} \log(\beta+\beta\sqrt{2\pi e}  )}{ \beta (4-2^{\alpha+1/2})} .\label{f_N-bound}
\ee
Note that the infinite series appearing above converges to a finite value for $1<\alpha<3/2$.
Moreover, we have
\be
f_1(1)&=&\frac{1}{2^{2}} \mathbb{E}[ \langle (\sigma_1 +\sigma_2)^2  \rangle_1] 
\no\\
&=&\frac{1}{2}+\frac{1}{2} \mathbb{E}[  \langle  \sigma_1 \sigma_2  \rangle_1] . \label{f_1}
\ee
Combining Eqs.~\eqref{m^r-f_N}, \eqref{f_N-bound}, and~\eqref{f_1}, we arrive at the statement of Theorem~\ref{theorem1}.
\end{proof}

In the remainder of this section, we provide proofs of Lemmas~\ref{lemma5}, \ref{lemma6}, and~\ref{lemma7}.

\subsection{Proof of Lemma \ref{lemma5}}\label{proof-lemma5}

The Gibbs--Bogoliubov inequality on the Nishimori line~\cite{OO-GB} implies that
\be
P_N(x_N,x_{N-1},\cdots,x_1)&\le& P_N(0,x_{N-1},\cdots,x_1)  + \frac{1}{2} x_{N}   (2^{2N} +\mathbb{E}[\langle S_{N,1}^2 \rangle_N])
\no\\
&=& 2P_{N-1}(x_{N-1},\cdots,x_1)  +   \frac{1}{2}\frac{\beta^2 b_N}{2^{2N}} (2^{2N} +\mathbb{E}[\langle S_{N,1}^2 \rangle_N]), \label{GB-upper}
\ee
where we used the recursive structure of the Dyson hierarchical model.
On the other hand, applying the Gibbs--Bogoliubov inequality on the Nishimori line in the opposite direction yields
\be
2P_{N-1}(2x_N+x_{N-1},x_{N-2},\cdots,x_1) 
&\ge&2P_{N-1}(x_{N-1},x_{N-2},\cdots,x_1) + 2 \frac{2x_N}{2} (2^{2(N-1)} +\mathbb{E}[\langle S_{N-1,1}^2 \rangle_{N-1}]) 
\no\\
&\ge&2P_{N-1}(x_{N-1},x_{N-2},\cdots,x_1) + 2 \frac{\beta^2 b_N}{2^{2N}} (2^{2(N-1)} +\mathbb{E}[\langle S_{N-1,1}^2 \rangle_{N-1}]) .\label{GB-lower}
\no\\
\ee
Combining inequalities~\eqref{GB-upper} and~\eqref{GB-lower}, and using the definition of the long-range order parameter~\eqref{def-fN}, we obtain the recurrence relation stated in Lemma~\ref{lemma5}.

\subsection{Proof of Lemma \ref{lemma6}}\label{proof-lemma6}
Using the fundamental theorem of calculus together with Gaussian integration by parts, we obtain
\be
Q_N(1)&=&Q_N(0)+ \int_0^1dt \dv{}{t} Q_N(t)
\no\\
&=&Q_N(0)- \frac{\beta^2 b_N}{2^{2N}} \int_0^1dt \mathbb{E}[  \langle  (S_{N-1,1}-S_{N-1,2})^2 \rangle_t'  -  \frac{1}{2} \langle  (q_{N-1,1}-q_{N-1,2})^2 \rangle_t' ]
\no\\
&\ge&Q_N(0)- \frac{\beta^2 b_N}{2^{2N}} \int_0^1dt \mathbb{E}[  \langle  (S_{N-1,1}-S_{N-1,2})^2 \rangle_t ' ].
\ee 
Here  $\langle \cdots\rangle_t'$ denotes the thermal average with respect to the interpolating Hamiltonian $H_{N,t}(\vec{\sigma})$, with the restricted sum over spin configurations
\be
\sum_{k=1}^{r_N} \Tr_{ \{ S_{N-1,1}\in I_k\}} \Tr_{ \{ S_{N-1,2}\in I_k\}}.
\ee
Moreover, the overlap variables are defined by
\be
q_{N-1,1}&=&\sum_{i=1}^{2^{N-1}} \sigma_i^1\sigma_i^2,
\\
q_{N-1,2}&=&\sum_{i=2^{N-1}+1}^{2^{N}} \sigma_i^1\sigma_i^2 ,
\ee
where the variables $\sigma_i^1$ and $\sigma_i^2$ represent the spins at site $i$ in the first and second replicas, respectively.
Using the bound~\eqref{res-S1-S2}, we deduce
\be
Q_N(1)
&\ge&Q_N(0)- 1
\no\\
&=&\mathbb{E}[\log (\sum_{k=1}^{r_N} Z_{N,k}(\{J_1,K_1\}) Z_{N,k}(\{J_2,K_2\})) ] -1.
\ee
Combining this estimate with Eq.~\eqref{P->Q} and applying the Cauchy--Schwarz inequality, we obtain
\be
P_N(x_N,x_{N-1},\cdots,x_1) 
&\ge& Q_N(1)
\no\\
&\ge&\mathbb{E}[\log \left(\sum_{k=1}^{r_N} Z_{N,k}(\{J_1,K_1\}) \right)^2 ] -\mathbb{E}[\log \sum_{k=1}^{r_N} \frac{Z_{N,k}(\{J_1,K_1\})  } { Z_{N,k}(\{J_2,K_2\})  }] -1
\no\\
&=&2\mathbb{E}[\log \sum_{k=1}^{r_N} Z_{N,k}(\{J_1,K_1\})  ] -\mathbb{E}[\log \sum_{k=1}^{r_N} \frac{Z_{N,k}(\{J_1,K_1\})  } { Z_{N,k}(\{J_2,K_2\})  }] -1
\no\\
&=&2P_{N-1}(2x_N+x_{N-1},x_{N-2},\cdots,x_1) -\mathbb{E}[\log \sum_{k=1}^{r_N} \frac{Z_{N,k}(\{J_1,K_1\})  } { Z_{N,k}(\{J_2,K_2\})  }]  -1,
\ee
where, in the last equality, we used the definition of  \(Z_{N,k}(\{J_1,K_1\})\) in Eq.~\eqref{Z_{N,k}} together with the reproduction property of Gaussian distributions.
This completes the proof of Lemma~\ref{lemma6}.

\subsection{Proof of Lemma \ref{lemma7}}\label{proof-lemma7}
In this subsection only, we rescale the interactions by a change of variables so that they are expressed in terms of standard Gaussian random variables.
Specifically, we write
\be
J_{ij}^{(N-1)} &=& \sqrt{\frac{ b_{N-1}}{2^{2 (N-1)} }} L_{ij}^{(N-1)} +\frac{\beta b_{N-1}}{2^{2 (N-1)} },
\\
 L_{ij}^{(N-1)} &\stackrel{\text{iid}}{\sim}& \mathcal{N}\left( 0, 1\right).
\ee
We consider the random variable
\be
g_k \equiv\log  \frac{Z_{N,k}(\{J_1,K_1\}) }{Z_{N,k}(\{J_2,K_2\}) },
\ee
viewed as a function of all normalized Gaussian random variables $\{ L_{ij}^{(p)}\}$.
This function is Lipschitz continuous with Lipschitz constant
\be
C_N=(\beta^2 2^N R_N)^{1/2},
\ee
where
\be
R_N&=&\sum_{p=1}^N 2^{-p}b_p
=\frac{2(1-2^{(1-\alpha)N})}{2^\alpha-2} .
\ee
Indeed, for any $L_{ij}^{(p)}$, we have
\be
\left| \frac{\partial g_k}{\partial L_{ij}^{(p)} } \right| \le \beta \sqrt{\frac{ b_{p}}{2^{2 (p)} } }.
\ee
For the Gaussian variables associated with the additional couplings $K_{ij}^{(N-1)}(0)$, the corresponding contribution is included as the $p=N$ term in the definition of $C_N$.
Therefore, Gaussian concentration for Lipschitz functions, namely the Tsirelson--Ibragimov--Sudakov inequality~\cite{BLM}, implies that for any $t>0$
\be
\Pr( g_k- \mathbb{E}[g_k]\ge t) &\le& \exp( - \frac{t^2}{2C_N^2}).
\ee
For any $\gamma>0$, we then obtain~\cite{AK}
\be
\mathbb{E}[ e^{\gamma (g_k- \mathbb{E}[g_k])} ]
&=&\gamma \int_{-\infty}^0 dt e^{\gamma t}\Pr(g_k- \mathbb{E}[g_k]\ge t)  +\gamma \int_{0}^\infty dt e^{\gamma t}\Pr(g_k- \mathbb{E}[g_k]\ge t) 
\no\\
&\le&\gamma \int_{-\infty}^0 dt e^{\gamma t}  +\gamma \int_{0}^\infty dt e^{\gamma t} e^{ - \frac{t^2}{2C_N^2}}
\no\\
&=& 1+\sqrt{2\pi} \gamma C_N e^{\frac{\gamma^2C^2}{2}}. \label{exp-error}
\ee
As a consequence,
\be
\mathbb{E}[\max_k  g_k] -\max_k \mathbb{E}[  g_k]
&\le&\mathbb{E}[\max_k (g_k-\mathbb{E}[  g_k])] 
\no\\
&=& \frac{1}{\gamma}\log \exp(\gamma\mathbb{E}[\max_k  (g_k-\mathbb{E}[  g_k])] )
\no\\
&\le& \frac{1}{\gamma} \log \mathbb{E}[e^{  \gamma \max_k(g_k-\mathbb{E}[ g_k])} ] 
\no\\
&\le& \frac{1}{\gamma}\log \sum_k \mathbb{E}[e^{  \gamma(g_k-\mathbb{E}[  g_k])} ] 
\no\\
&\le& \frac{1}{\gamma}\log \left( r_N \left(1+\sqrt{2\pi} \gamma C_N e^{\frac{\gamma^2C_N^2}{2}} \right)\right) ,
\ee
where we used Jensen's inequality in the second inequality and Eq.~\eqref{exp-error} in the last inequality.
Choosing $\gamma=1/C_N=(\beta^2 2^N R_N)^{-1/2}$, we finally obtain
\be
&&\mathbb{E}[ \max_k \log  \frac{Z_{N,k}(\{J_1,K_1\})  } { Z_{N,k}(\{J_2,K_2\})  }]
\no\\
&\le&   \max_k\mathbb{E}[  \log  \frac{Z_{N,k}(\{J_1,K_1\})  } { Z_{N,k}(\{J_2,K_2\})  }]  +\beta 2^{N/2}R_N^{1/2}\log \left(r_N (1+\sqrt{2\pi e}  ) \right)
\no\\
&=&   \max_k \left\{\mathbb{E}[  \log  Z_{N,k}(\{J_1,K_1\})  - \log Z_{N,k}(\{J_2,K_2\}) ] \right\}  +\beta 2^{N/2}R_N^{1/2}\log\left( r_N (1+\sqrt{2\pi e}  )\right) 
\no\\
&=&  0+ \beta 2^{N/2}R_N^{1/2}\log \left(r_N (1+\sqrt{2\pi e}  ) \right),
\ee
which proves Lemma~\ref{lemma7}.

\section{Proof of Theorem \ref{theorem2}}
Owing to the hierarchical structure of the Dyson hierarchical lattice, for each pair of sites $(i,j)$ there exists a unique integer $p$  such that the spins $\sigma_i$ and $\sigma_j$ belong to a common block $S_{p,r}$ but not to any common block $S_{p-1,r}$.
In other words, the pair of sites $(i,j)$ first interacts at the $p$-th hierarchical level.
As a consequence, their distance is bounded by
\be
|i-j| <2^{p}. \label{distance-p}
\ee
From the definition of the Dyson hierarchical Hamiltonian~\eqref{Dyson-H}, the total interaction between the pair $(i,j)$ is given by
\be
\sum_{q=p}^N   ( J_{ij}^{(q)} +J_{ji}^{(q)}).
\ee
By the reproduction property of Gaussian distributions together with Eq.~\eqref{Jij-Dyson}, we have the equality in distribution
\be
\sum_{q=p}^N ( J_{ij}^{(q)} +J_{ji}^{(q)}) \label{}
&\stackrel{d}{=}& J_{ij}'  ,
\ee
where $J_{ij}'$ is a Gaussian random variable distributed as
\be
J_{ij}' &\sim& \mathcal{N}\left(\beta R_N(p),R_N(p) \right), \label{J'ij}
\ee
with
\be
R_N(p)&=& 2\sum_{q=p}^N \frac{b_q}{2^{2q}}
\no\\
&\le&2\sum_{q=p}^N\frac{1}{2^{q}} 2^{(1-\alpha)p}
\no\\
&\le& 2^{2-\alpha p}
\no\\
&<&\frac{4}{|i-j|^{\alpha}} . \label{J'ij-bound}
\ee
In the last inequality, we used Eq.~\eqref{distance-p}.
Equations~\eqref{J'ij} and~\eqref{J'ij-bound} show that, for any pair of sites $(i,j)$, the effective interaction strength in the Dyson hierarchical Ising spin glass model on the Nishimori line is smaller than that in the one-dimensional Ising spin glass model with long-range interactions on the Nishimori line defined by Eq.~\eqref{Jij-def}, provided that the two systems have the same size.
Therefore, the Griffiths inequality on the Nishimori line~\eqref{Griffiths-NL} implies that, for any pair of sites $(k,l)$,
\be
\mathbb{E}[\langle \sigma_k \sigma_l  \rangle_N] &\le& \mathbb{E}[\langle \sigma_k \sigma_l \rangle_{\text{long}}],
\ee
where $\langle \cdots\rangle_{\text{long}}$ denotes the thermal average with respect to the one-dimensional Ising spin glass model with long-range interactions on the Nishimori line defined by Eq.~\eqref{long-H}.
This completes the proof of Theorem~\ref{theorem2}.

\section{Proof of Theorem \ref{theorem3}}

For any fixed pair of sites $(i,j)$, we introduce an interpolating parameter $0\le t\le1$ into the Gaussian distribution of all interactions connected to the site $j$ by replacing
\be
J_{jk} \to J_{jk}(t) \sim \mathcal{N}\left(t  \frac{4\beta}{ |j-k|^{\alpha}}, t \frac{4}{|j-k|^{\alpha}} \right) .
\ee
We denote the corresponding $t$-dependent correlation function by $\mathbb{E}[\langle \sigma_i\sigma_j \rangle_t]$.
Note that, at $t=1$, the correlation function coincides with that of the original model,
\be
\mathbb{E}[\langle \sigma_i\sigma_j \rangle_1]&=&\mathbb{E}[\langle \sigma_i\sigma_j \rangle_{\text{long}}],\label{correlation-t=1}
\ee
while at $t=0$ the interaction between site $j$ and all other sites vanishes, and hence
\be
\mathbb{E}[\langle \sigma_i\sigma_j \rangle_0]&=&\mathbb{E}[\langle \sigma_i \rangle_0 \langle\sigma_j \rangle_0]
=\mathbb{E}[\langle \sigma_i \rangle_0 \cdot 0]=0.
\ee
Following Ref.~\cite{OO-NL-bound}, we compute the derivative with respect to $t$
\be
\dv{}{t}\mathbb{E}[\langle \sigma_i \sigma_j \rangle_t]
&=&\sum_{k(\neq j)}x_{jk} \mathbb{E}[\langle \sigma_i \sigma_k \rangle_t^2 + \langle \sigma_i \sigma_j \rangle_t^2 \langle \sigma_j \sigma_k \rangle_t^2-2\langle \sigma_i \sigma_j \rangle_t \langle \sigma_i \sigma_k \rangle_t \langle \sigma_j \sigma_k \rangle_t ]
\no\\
&=&\sum_{k(\neq j)}x_{jk} \mathbb{E}[\langle \sigma_i \sigma_k \rangle_t^2 + \langle \sigma_i \sigma_j \rangle_t^2 \langle \sigma_j \sigma_k \rangle_t^2-2\langle \sigma_i \sigma_j \rangle_t \langle \sigma_i \sigma_k \rangle_t  ]
\no\\
&\le&\sum_{k(\neq j)}x_{jk} \mathbb{E}[\langle \sigma_i \sigma_k \rangle_t^2 + \langle \sigma_i \sigma_j \rangle_t^2 \langle \sigma_j \sigma_k \rangle_t^2+ \langle \sigma_i \sigma_j \rangle_t^2 +  \langle \sigma_i \sigma_k \rangle_t^2  ] ,
\ee
where $x_{jk}=4\beta^2 /|j-k|^{\alpha}$.
In the second equality, we used the identity on the Nishimori line~\cite{Nishimori2}
\be
\mathbb{E}[\langle \sigma_i \sigma_j \rangle_t \langle \sigma_i \sigma_k \rangle_t \langle \sigma_j \sigma_k \rangle_t]= \mathbb{E}[\langle \sigma_i \sigma_j \rangle_t \langle \sigma_i \sigma_k \rangle_t],
\ee
and in the inequality we applied  $-2xy\le x^2+y^2$.
Using $\langle \sigma_j \sigma_k \rangle_t^2\le1$ and the Nishimori identity~\eqref{Nishimori-identity}, we further obtain
\be
\dv{}{t}\mathbb{E}[\langle \sigma_i \sigma_j \rangle_t]
&\le&\sum_{k(\neq j)}x_{jk} \mathbb{E}[\langle \sigma_i \sigma_k \rangle_t^2 + \langle \sigma_i \sigma_j \rangle_t^2 + \langle \sigma_i \sigma_j \rangle_t^2 +  \langle \sigma_i \sigma_k \rangle_t^2  ]
\no\\
&=&2\sum_{k(\neq j)}x_{jk} \mathbb{E}[ \langle \sigma_i \sigma_k \rangle_t + \langle \sigma_i \sigma_j \rangle_t ]
\no\\
&\le&2\sum_{k(\neq j)}x_{jk} \mathbb{E}[ \langle \sigma_i \sigma_k \rangle_{\text{long}} + \langle \sigma_i \sigma_j \rangle_{\text{long}} ],
\ee
where we used Eq. \eqref{correlation-t=1} and the Griffiths inequality on the Nishimori line~\eqref{Griffiths-NL} in the last inequality.
By integrating with respect to $t$ from $0$ to $1$, and using $\mathbb{E}[\langle \sigma_i \sigma_j \rangle_{1}]=\mathbb{E}[\langle \sigma_i \sigma_j \rangle_{\text{long}}]$,
we obtain
\be
\mathbb{E}[\langle \sigma_i \sigma_j \rangle_{\text{long}}] 
&\le& \mathbb{E}[\langle \sigma_i \sigma_j \rangle_0] +2\sum_{k(\neq j)}x_{jk} \mathbb{E}[ \langle \sigma_i \sigma_k \rangle_{\text{long}} + \langle \sigma_i \sigma_j \rangle_{\text{long}} ]
\no\\
&=& 2\sum_{k(\neq j) }x_{jk} \mathbb{E}[ \langle \sigma_i \sigma_k \rangle_{\text{long}}  ] +2\sum_{k(\neq j)  }x_{jk} \mathbb{E}[  \langle \sigma_i \sigma_j \rangle_{\text{long}} ].
\ee
Summing over $i(\neq j)$ we obtain
\be
\sum_{i(\neq j)}\mathbb{E}[\langle \sigma_i \sigma_j \rangle_{\text{long}}] 
&\le&  2  \sum_{k(\neq j) }x_{jk} \sum_{i(\neq j)} \mathbb{E}[ \langle \sigma_i \sigma_k \rangle_{\text{long}}  ] +2 \sum_{k(\neq j)  }x_{jk} \sum_{i(\neq j)}\mathbb{E}[  \langle \sigma_i \sigma_j \rangle_{\text{long}} ]
\no\\
&\le&  2  \sum_{k(\neq j) }x_{jk} \sum_{i} \mathbb{E}[ \langle \sigma_i \sigma_k \rangle_{\text{long}}  ] +2 \sum_{k(\neq j)  }x_{jk} \sum_{i(\neq j)}\mathbb{E}[  \langle \sigma_i \sigma_j \rangle_{\text{long}} ]
\no\\
&=&  2  \sum_{k(\neq j) }x_{jk}\left( \sum_{i(\neq k)} \mathbb{E}[ \langle \sigma_i \sigma_k \rangle_{\text{long}}  ] +1 \right) +2 \sum_{k(\neq j)  }x_{jk} \sum_{i(\neq j)}\mathbb{E}[  \langle \sigma_i \sigma_j \rangle_{\text{long}} ]
\no\\
&\le& 2  \sum_{k(\neq j) }x_{jk} \left(\max_j \sum_{i(\neq j)} \mathbb{E}[ \langle \sigma_i \sigma_j \rangle_{\text{long}}  ] +1 \right)+2 \sum_{k(\neq j)  }x_{jk} \sum_{i(\neq j)}\mathbb{E}[  \langle \sigma_i \sigma_j \rangle_{\text{long}} ] .
\ee
Recalling that $x_{jk}=4\beta^2 /|j-k|^{\alpha}$, we have
\be
\sum_{k(\neq j) }x_{jk} <  8\beta^2 M_0,
\ee
where 
\be
M_0&\equiv&\sum_{i=1}^{\infty}\frac{1}{ |i|^{\alpha}} <  \infty.
\ee
Therefore,
\be
\sum_{i(\neq j)}\mathbb{E}[\langle \sigma_i \sigma_j \rangle_{\text{long}}] 
&\le& 16  \beta^2M_0 \left(\max_j \sum_{i(\neq j)} \mathbb{E}[ \langle \sigma_i \sigma_j \rangle_{\text{long}}  ] +1 \right)+16\beta^2M_0 \sum_{i(\neq j)}\mathbb{E}[  \langle \sigma_i \sigma_j \rangle_{\text{long}} ].
\ee
Taking the maximum over $j$ on both sides, we arrive at
\be
(1-32\beta^2M_0) \max_j \sum_{i(\neq j)}\mathbb{E}[\langle \sigma_i \sigma_j \rangle_{\text{long}}] 
&\le& 16\beta^2M_0 .
\ee
Hence, if
\be
1 > 32\beta^2 M_0,
\ee
the infinite sum
\be
\sum_{i (\ne j)} \mathbb{E}\bigl[\langle \sigma_i \sigma_j \rangle_{\text{long}} \bigr]
\ee
converges absolutely for any fixed $j$. Therefore,
\be
\mathbb{E}\bigl[\langle \sigma_i \sigma_j \rangle_{\text{long}} \bigr] \to 0
\ee
as $|i-j| \to \infty$.
Consequently, there is no long-range order for $1 > 32\beta^2 M_0$, and Theorem~\ref{theorem3} follows.

\section{Discussions}\label{Discussions}
Before the present work, to the best of our knowledge, there had been no rigorous results establishing the existence of a phase transition for one-dimensional Ising spin glass models with long-range interactions in the region $1<\alpha<2$.
In this work, we have rigorously proved the existence of a phase transition for  $1<\alpha<3/2$  in the one-dimensional Ising spin glass model with long-range interactions on the Nishimori line.
Although ferromagnetic order at low temperatures on the Nishimori line is expected due to the strong ferromagnetic bias, establishing this rigorously remains highly nontrivial.

The reason why our proof does not extend beyond $\alpha=3/2$ is as follows.
In the proof of Lemma~\ref{lemma7}, the use of a probability concentration inequality produces an error term of order $\mathcal{O}(2^{N/2})$.
For $3/2\le \alpha$, this error term prevents the infinite sum appearing in Eq.~\eqref{f_N-inf-sum} from converging to a finite value.
As a result, we are unable to obtain a finite lower bound on the long-range order parameter in the Dyson hierarchical Ising spin glass model on the Nishimori line.
Determining whether long-range order exists in the region $3/2\le \alpha\le2$ therefore remains an important open problem.
Although we cannot prove this rigorously, the behavior of the ferromagnetic Ising model and the expected scenario for long-range spin glasses make it plausible that long-range order persists at least for $3/2\le\alpha<2$~\cite{KS2,EH,KAS,KY,Moore,MG,APT,LPTL}.

On the other hand, when our method is applied to the random-field Ising model on the Dyson hierarchical lattice, the existence of a phase transition for $1<\alpha<3/2$ can be established in an analogous manner~\cite{OO-Dyson}.
Interestingly, in this case it is considered that no phase transition occurs for $3/2<\alpha$~\cite{AW}, and our approach is consistent with this picture.

We remark that our method is restricted to Gaussian interactions on the Nishimori line.
Although the Nishimori line can be defined for more general distributions, such as binary ones, the present approach does not extend to these cases.
In particular, the interpolation method is incompatible with binary disorder; the Gibbs--Bogoliubov inequality on the Nishimori line fails in this setting~\cite{OO-GB}; and the reproduction property of the distribution, which is crucial for the proof of Theorem~\ref{theorem2}, does not hold.
Extending the present analysis beyond the Gaussian case therefore remains a challenging problem for future work.

We finally comment on critical exponents. 
For the ferromagnetic one-dimensional Ising model with long-range interactions, reflection positivity, together with the infrared bound and differential inequalities,
can be used to prove mean-field critical exponents in the range $1<\alpha<3/2$~\cite{AF}. 
For the present spin-glass model on the Nishimori line, however, reflection positivity is not available because the interactions are random and can take both signs. 
Therefore, neither our approach nor, to the best of our knowledge, the existing methods appear to yield rigorous information on critical exponents for the model considered in this paper. 
Critical exponents in spin-glass models with long-range interactions remain an important and active topic~\cite{JRP,FPT}, and we leave this problem for future work.

\section*{Data availability statement}
No new data were created or analyzed in this study.

\section*{Acknowledgment}
This work was supported by JST BOOST, Japan Grant Number JPMJBY24B6.
This work was also supported by JSPS KAKENHI Grant Nos. 24K16973 and 23H01432.
In addition, this work was supported by programs for bridging the gap between R\&D and IDeal society (Society 5.0) and Generating Economic and social value (BRIDGE)
and Cross-ministerial Strategic Innovation Promotion Program (SIP) from the Cabinet Office (No. 23836436).


\end{document}